\documentclass{article}
\pdfoutput=1

\usepackage{tikz}
\usepackage{graphicx,amssymb,amsmath}
\usepackage{amsthm}
\usepackage{xspace}
\usepackage{thmtools, thm-restate}
\usepackage[hidelinks]{hyperref}

\usepackage{xcolor}
\usepackage{subcaption}

\usepackage{lineno}

\graphicspath{{figures/}}

\newcommand{\V}{\mathsf{v}}
\renewcommand{\H}{\mathsf{h}}

\newcommand{\la}{\langle}
\newcommand{\ra}{\rangle}
\newcommand{\df}[1]{\textit{#1}} 

\newcommand{\eps}{\epsilon}

\newcommand{\phprim}{S_\pi}
\newcommand{\pvprim}{W_\pi}
\newcommand{\projx}{X_\G}
\newcommand{\projy}{Y_\G}

\renewcommand{\b}{\underline{y}_\G}
\newcommand{\bv}{\underline{y}_{\G_\V}}
\newcommand{\bo}{\b}
\renewcommand{\t}{\overline{y}_\G}
\renewcommand{\l}{\underline{x}_\G}
\newcommand{\lh}{\underline{x}_{\G_\H}}
\renewcommand{\r}{\overline{x}_\G}

\newcommand{\SW}{\texttt{SW}\xspace}
\newcommand{\NW}{\texttt{NW}\xspace}
\newcommand{\SE}{\texttt{SE}\xspace}
\newcommand{\NE}{\texttt{NE}\xspace}

\newcommand{\Gp}{\G'}
\newcommand{\Gvp}{\G'_\V}
\newcommand{\Ghp}{\G'_\H}
\newcommand{\bp}{\underline{y}_{\G'}}
\newcommand{\tp}{\overline{y}_{\G'}}
\newcommand{\lp}{\underline{x}_{\G'}}
\newcommand{\rp}{\overline{x}_{\G'}}

\newcommand{\pv}{P_\V}
\newcommand{\ph}{P_\H}
\newcommand{\pvph}{\la \pv,\ph\ra}
\newcommand{\gvgh}{\la G_\V,G_\H\ra}

\newcommand{\G}{\Gamma}

\newcommand{\Gv}{\G_\V}
\newcommand{\Gh}{\G_\H}

\newcommand{\gv}{G_\V}
\newcommand{\gh}{G_\H}
\newcommand{\lsvrpaths}{\textsf{LsvrPaths}\xspace}

\newcommand{\NP}{\textsf{NP}}

\renewcommand{\llcorner}{\tikz\draw (0,0) -- ++(0,-1.5ex) -- ++(1.5ex,0);}
\renewcommand{\lrcorner}{\tikz\draw (0,0) -- ++(1.5ex,0) -- ++(0ex,1.5ex);}
\renewcommand{\ulcorner}{\tikz\draw (0,0) -- ++(-1.5ex,0) -- ++(0,-1.5ex);}
\renewcommand{\urcorner}{\tikz\draw (0,0) -- ++(1.5ex,0) -- ++(0,-1.5ex);}
\newcommand{\lset}{\{\llcorner,\lrcorner,\urcorner,\ulcorner\}}
\newcommand{\spider}{\tikz[scale=0.2]{%
\draw (0,0) -- (-0.5,-0.8660254) -- (-1,0); %
\draw (0,0) -- (-0.5,0.8660254) -- (0.5,0.8660254);%
\draw (0,0) -- (1,0) -- (0.5,-0.8660254);%
}}

\newcommand{\add}[1]{{#1}}

\newcommand{\comment}[1]{}

\newtheorem{lemma}{Lemma}
\newtheorem{theorem}{Theorem}

\newtheorem{corollary}{Corollary}
\newtheorem{definition}{Definition}
\newtheorem{claim}{Claim}

\title{Simultaneous Visibility Representations of Undirected Pairs of Graphs\thanks{Supported by Canada NSERC Discovery Grant and Undergraduate Student Research Awards. }}

\author{Ben Chugg\thanks{Stanford University, United States, \texttt{benchugg@stanford.edu}}
\and
William S. Evans\thanks{University of British Columbia, Canada, \texttt{will@cs.ubc.ca}}
\and
Kelvin Wong\thanks{University of Toronto, Canada, \texttt{kelvinwong@cs.toronto.edu}}}

\index{Chugg, Ben}
\index{Evans, William S.}
\index{Wong, Kelvin}
\date{}

\begin{document}
\thispagestyle{empty}
\maketitle

\begin{abstract}
We consider the problem of determining if a pair of undirected graphs $\la G_\V, G_\H \ra$,
which share the same vertex set, has a representation using opaque geometric shapes for vertices,
and vertical (respectively, horizontal) visibility between shapes to determine the edges of $G_\V$ (respectively, $G_\H$).
While such a simultaneous visibility representation of two graphs can be determined efficiently if the direction of the required visibility for each edge is provided (and the vertex shapes are sufficiently simple), it was unclear if edge direction is critical for efficiency.
Here, an edge directed from $u$ to $v$ implies that the shape representing $u$ is below (respectively, left of) the shape for $v$ in $G_\V$ (respectively, $G_\H$).
We show that the problem is \NP-complete without that information, even for graphs that are only slightly more complex than paths.
In addition, we characterize which pairs of paths have simultaneous visibility representations using fixed orientation L-shapes.
This narrows the range of possible graph families for which determining simultaneous visibility representation is non-trivial yet not \NP-hard.
\end{abstract}


\section{Introduction}
A \emph{visibility representation} $\G$ of a graph $G = (V,E)$ is a set of
disjoint geometric objects $\{ \G(v) \;|\; v \in V \}$ representing vertices
chosen from a family of allowed objects (e.g., axis-aligned
rectangles in the plane) where $\G(u)$ sees $\G(v)$ if and only if $uv \in E$.
Typically, the meaning of ``sees'' is that there exists a line segment
(perhaps axis-aligned, perhaps positive width)
from $\G(u)$ to $\G(v)$ that does not intersect
$\G(w)$ for any other $w \in V$; such a line segment is called a
\emph{line-of-sight}.
Many different classes of visibility representations may be defined by
changing the family of allowed objects and the meaning of ``sees.''
For example,
\emph{bar visibility representations} (BVRs) use horizontal line segments as vertices and vertical lines-of-sight for edges~\cite{dean2003unit,duchet83,wismath85,tamassia86,dean2007bar,felsner2008parameters,hartke2007further};
\emph{rectangle visibility representations} (RVRs) use (solid) axis-aligned rectangles and axis-aligned lines-of-sight~\cite{dean1994rectangle,wismath89,shermer96,bose96,cobos1995visibility,hutchinson1999representations};
and
\emph{unit square visibility representations} (USVRs) use axis-aligned unit squares and axis-aligned lines-of-sight~\cite{casel17}.
The popularity of this type of graph representation lies in its
potential applicability to problems in VLSI design and the production
of readable representations of planar and non-planar graphs.
Determining which graphs or families of graphs have visibility
representations of a particular type is a fascinating area of
research.

\begin{figure}[t]
\centering
\begin{subfigure}[t]{0.4\textwidth}
\centering
\includegraphics{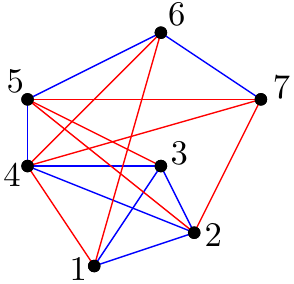}
\end{subfigure}
\begin{subfigure}[t]{0.4\textwidth}
\centering
\includegraphics{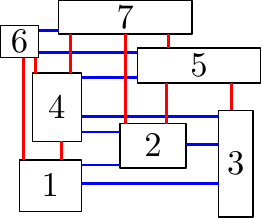}
\end{subfigure}
\caption{Left: A pair of graphs $\gvgh$ on the same vertex set. Red (light) edges are those of $\gv$ while blue (darker) are those of $\gh$. Right: A simultaneous visibility representation of $\gvgh$ using rectangles. 
}
\label{fig:svr_example}
\end{figure}

\subsection{Contribution}

Our focus in this work is on the \emph{simultaneous} visibility
representation of pairs of graphs that
share the same vertex set (see Fig.~\ref{fig:svr_example}).
A \emph{simultaneous visibility representation} (SVR) of $G_\V=(V,E_\V)$ and
$G_\H=(V,E_\H)$ is a visibility representation $\G$ that,
using vertical lines-of-sight,
represents $G_\V$
and,
using horizontal lines-of-sight,
represents $G_\H$.
Streinu and Whitesides~\cite{streinu03} describe a beautiful connection between a pair of \emph{directed} planar graphs $\la \overrightarrow{G}_\V,
\overrightarrow{G}_\H \ra$ and their planar duals that determines if the pair has a directed simultaneous visibility representation using rectangles (a directed RSVR) $\G$,
where an edge directed from $u$ to $v$ in $\overrightarrow{G}_\V$ or $\overrightarrow{G}_\H$
is realized by a
low-to-high or left-to-right, respectively, line-of-sight from $\G(u)$ to $\G(v)$.
Evans et al.~\cite{evans16} extended this to the family of
geometric objects called 
\emph{L-shapes},
which are the union of two axis-aligned segments in the plane that
share a common endpoint and come in four orientations: $\lset$.
They gave a polynomial time algorithm for determining if a pair of
directed graphs $\la \overrightarrow{G}_\V,
\overrightarrow{G}_\H \ra$
has a directed simultaneous visibility
representation using L-shapes $\G$
in which the orientation of $\G(v)$ is given by $\Phi:V \rightarrow
\lset$ (a directed $\Phi$-LSVR).

The complexity of determining if a pair of \emph{undirected} graphs
has a simultaneous visibility representation using L-shapes
was stated as an open problem~\cite{evans16}.
In this paper, we show (Section~\ref{sec:np-complete}) that the problem
is \NP-complete.
What is surprising about this result is the simplicity of the graphs
for which the problem is hard:
For L-shapes (and many other families of shapes including rectangles),
one graph can be a set of disjoint paths
and the other a set of disjoint copies of a tree with 3 leaves connected to a root by paths of length 2 (i.e. $\spider$).
For unit squares
(and for translates of any specified connected shape with positive width and height),
one graph can
be a set of disjoint paths and the other a set of disjoint claws ($K_{1,3}$).

This limits the families of graphs for which we can reasonably hope to
efficiently determine a simultaneous visibility representation.
We describe a linear time algorithm (Section~\ref{sec:paths}) that
determines if a pair of undirected \emph{paths} has
a simultaneous visibility representation using L-shapes, all with the
same orientation $\{\llcorner\}$ (an LSVR).
The algorithm is quite simple but relies on characterizing those pairs
of paths for which such a representation is possible.
The characterization of such pairs of paths for representations using
rectangles, unit squares, and L-shapes with more than one orientation is easier.

\subsection{Related Work}

Our work has aspects of both visibility representation and simultaneous geometric graph embedding (SGE).
SGE is the problem of deciding, given a set of planar graphs on the same set of vertices, whether the vertices can be placed in the plane so that each graph has a straight-line drawing on the placed vertices.
As in our problem, SGE, which is \NP-hard~\cite{estrella07}, asks to represent several specified graphs using one common vertex set representation.
However, the hardness result for SGE does not directly imply hardness of deciding simultaneous visibility representation.
Similarly, deciding if a graph has an RVR~\cite{shermer96} or a USVR~\cite{casel17} is \NP-hard, but since the input does not specify which edges should be realized as vertical versus horizontal lines-of-sight, the problems are quite different.
Choosing how the graph should be split into vertical and horizontal parts is an additional opportunity (or burden) for deciding if these representations exist.

Rather than requiring the visibility representation to partition the edges of the graph in a prescribed manner between vertical and horizontal visibilities, Biedl et al.~\cite{biedl2018embedding} require that the visibility edges (lines-of-sight) obey the same embedding as a prescribed embedding of the original graph, which may include edge crossings.
They can decide if such a restricted RVR exists in polynomial time and in linear time if the graph is 1-planar.\footnote{\add{A graph is 1-planar if it has a planar embedding in which each edge crosses at most one other.}}
Di Giacomo et al.~\cite{ortho18} show that deciding if a similarly restricted ortho-polygon\footnote{a polygon whose edges are axis-aligned.} visibility representation exists 
for an embedded graph takes polynomial time as well.

\section{Preliminaries}
\label{sec:prelims}
In this paper, we will assume that vertex shapes are connected and
closed (rather than open) sets in the plane and that
lines-of-sight are 0-width (rather than positive-width) and exist
between two shapes if and only if the corresponding vertices are
connected by an edge.
This implies that $G_\V$ and $G_\H$ must have \emph{strong-visibility
  representations}~\cite{tamassia86} to have a simultaneous visibility representation.
For visibility graphs, these choices make a difference since, for
example, $K_{2,4}$ can be represented if lines-of-sight are
positive-width (an \emph{$\epsilon$-visibility representation}) but
does not have a strong-visibility representation~\cite{tamassia86}.
However, for our results, we could adopt either model with only minor
modifications to our proofs. 
In particular, we could allow the vertical and/or horizontal projection of a shape to be an interval that is closed or open on either end (the $\epsilon$-visibility model) rather than a segment that is closed on both ends (the strong-visibility model).  The proofs of Properties~\ref{prop:ThinOverlap} and \ref{prop:nextBarsSameSide} (below), for example, would remain the same.

Let $\G$ be an SVR of $\gvgh$. Given a subset $S \subset V$,
we let $\G(S) = \bigcup_{v \in S}\G(v)$.
For a vertex $v \in V$, let $X_\Gamma(v)$ and $Y_\Gamma(v)$ be the
orthogonal projections of $\Gamma(v)$ onto the $x$-axis and $y$-axis
respectively.
For a set of vertices $S \subset V$, let
$X_\Gamma(S) = \bigcup_{v  \in S} X_\Gamma(v)$ and
$Y_\Gamma(S) = \bigcup_{v  \in S} Y_\Gamma(v)$. 
Set $\l(v) = \min X_\G(v)$, $\r(v) = \max X_\G(v)$, $\b(v) = \min Y_\G(v)$, and $\t(v) = \max Y_\G(v)$. 
We write $Y_\G(u) \leq Y_\G(v)$ if $\t(u) \leq \b(v)$ and $\projx(u)\leq \projx(v)$ if $\r(u)\leq \l(v)$. We also use the shorthand $[n]$ for $\{1,2,\dots,n\}$.

\add{When discussing the relative positions of shapes in an SVR $\Gamma$, we use above/below (resp. right/left) when two shapes can be separated by a horizontal (resp. vertical) line.
}

We state two basic properties visibility representations $ \G $ of a graph $ G = (V, E) $. 
For brevity, these properties are stated for vertical visibility representations only but also hold for horizontal visibility representations by symmetry. 


\begin{restatable}{property}{ThinOverlap}
\label{prop:ThinOverlap}
For $S_1, S_2\subset V$ and \add{a point} $x \in X_\Gamma(S_1) \cap X_\Gamma(S_2)$, there exists a
path $u = u_1, \dots, u_k = v$ in $G$ 
for some $u \in S_1$ and $v \in S_2$ such that $x \in X_\Gamma(u_i)$ for all $i \in [k]$.
\end{restatable}
\begin{proof}
Consider the intersection of $\Gamma$ and the infinite vertical
line $x \times (-\infty,+\infty)$.
Since this line intersects both $\Gamma(S_1)$ and $\Gamma(S_2)$,
there must be two vertices $u \in S_1$ and $v \in S_2$ such that
$\Gamma(u)$ and $\Gamma(v)$ both intersect the line.
Let $u=u_0,u_1, \dots, u_k=v$ be the sequence of vertices in $V$ that
intersect the line in order along the line from $\Gamma(u)$ to
$\Gamma(v)$.  $\Gamma(u_i)$ and $\Gamma(u_{i+1})$ have an unblocked
vertical visibility segment between them for all $1\leq i < k$, which
implies a path between $u$ and $v$ that
connects $S_1$ and $S_2$ in $G$. 
\end{proof}

\begin{restatable}{property}{NextBarsSameSide}
\label{prop:nextBarsSameSide}
Let $ u_1, \ldots, u_\ell $ be the only path from $u_1$ to $u_\ell$ in $ G $.
If $ \Gamma(u_i) $ and $ \Gamma(u_k) $ are both above or both below $ \Gamma(u_j) $ for some $1 \leq i < j < k \leq \ell$,
then $ X_\Gamma(\{u_1, \ldots, u_i\}) \cap X_\Gamma(\{u_k, \ldots, u_\ell \}) = \emptyset $.
\end{restatable}
\begin{proof}
	If $ x \in X_\Gamma(\{u_1, \ldots, u_i\}) \cap X_\Gamma(\{u_k, \ldots, u_\ell \}) $, then by Property~\ref{prop:ThinOverlap}, there exists a path from $ u_i $ to $ u_k $ (following the vertical line through $ x $) that, since $ \Gamma(u_i) $ and $ \Gamma(u_k) $ are both above or both below $ \Gamma(u_j) $, does not include $ u_j $, a contradiction.
\end{proof}

\section{Hardness}
\label{sec:np-complete}

In this section, we study the complexity of determining if a pair of undirected graphs has an SVR.
We first consider the problem of determining SVRs using unit squares (Section~\ref{sec:ssvr-recognition}).
Then, we discuss how our results can generalize to other connected shapes as well (Section~\ref{sec:hardness-generalization}).
In the case of L-shapes, our results settle an open question of Evans et al.~\cite{evans16}.

\begin{figure}[t]
\centering
\begin{subfigure}[t]{0.45\textwidth}
\centering
\includegraphics{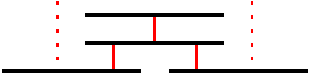}
\end{subfigure}
\hfill
\begin{subfigure}[t]{0.45\textwidth}
\centering
\includegraphics{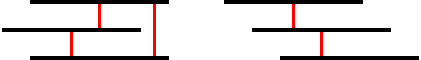}
\end{subfigure}
\caption[LOF caption]{Illustrations of
Lemma~\ref{le:Nestedness} (left):
A cut vertex will not \emph{nest}\protect{\footnotemark}~at
most two components induced by its removal;
and Lemma~\ref{le:NoTwist} (right):
Disjoint subgraphs must not overlap.}
\label{fi:hardness-lemma}
\end{figure}
\footnotetext{$\Gamma(A) $ is \emph{nested} in $ \Gamma(B) $ if $ X_\Gamma(A) \supseteq X_\Gamma(B) $}

To begin, we first state two lemmas that characterize how the gadgets in our hardness proofs can be drawn.
See Fig.~\ref{fi:hardness-lemma}
for illustrations of
the lemmas.

\begin{restatable}{lemma}{Nestedness}
\label{le:Nestedness}
Let $ G = (V, E) $ be a connected graph with a (vertical) visibility representation $\Gamma$.
If $ u \in V $ is a cut vertex whose removal
creates components $ C_1, \ldots, C_k $,
then $X_\Gamma(C_i) \not\subseteq X_\Gamma(u)$ for at most two
components.
\end{restatable}
\begin{proof}
Since $G$ is connected, $X_\Gamma(C_i)$ intersects $X_\Gamma(u)$
and if in addition $X_\Gamma(C_i) \not\subseteq X_\Gamma(u)$, then, since $C_i$ is connected, $X_\Gamma(C_i)$ is a contiguous interval
that strictly contains an endpoint of $X_\Gamma(u)$.
If three components have this property, then for two of
them, say $C_i$ and $C_j$, $X_\Gamma(C_i)$ and
$X_\Gamma(C_j)$ strictly contain the same endpoint and thus contain a
point $x \not\in X_\Gamma(u)$.
By Property~\ref{prop:ThinOverlap}, $G$
contains a path (following the vertical line through $x$) between $C_i$ and $C_j$ that does not contain $u$, a
contradiction.
\end{proof}

\begin{restatable}{lemma}{NoTwist}
\label{le:NoTwist}
Let $ G = (V, E) $ be a graph with a (vertical) visibility representation $ \Gamma $.
If $ C_1 $ and $ C_2 $ are 
components in $ G $, then either $ X_\Gamma(C_1) > X_\Gamma(C_2) $\footnote{$ X_\Gamma(A) > X_\Gamma(B) $ means $ X_\Gamma(a) > X_\Gamma(b) $ for all $ a \in A $, $ b \in B $}
\add{or $ X_\Gamma(C_2) > X_\Gamma(C_1) $.}
\end{restatable}
\begin{proof}
Since each component $ C_i $ is connected, its $ x $-projection $ X_\Gamma(C_i) $ forms a contiguous interval.
And yet since $ C_1 $ and $ C_2 $ are disconnected in $ G $, by Property~\ref{prop:ThinOverlap}, $ X_\Gamma(C_1) \cap X_\Gamma(C_2) = \emptyset $.
Thus either $ X_\Gamma(C_1) > X_\Gamma(C_2) $
\add{or $ X_\Gamma(C_2) > X_\Gamma(C_1) $.}
\end{proof}

\subsection{USSVR recognition}
\label{sec:ssvr-recognition}
We first prove that determining if a pair of undirected graphs has a \emph{simultaneous visibility representation using unit squares} (USSVR) is \NP-complete.


\begin{theorem}\label{thm:ssvr-recognition-is-hard}
Deciding if a pair of undirected graphs has a USSVR is \NP-complete.
\end{theorem}

For our proof, we reduce from the \NP-complete problem of Monotone Not-All-Equal 3SAT~\cite{schaefer78}.
This variant of 3SAT stipulates that every clause has three positive literals of which exactly one or two must be satisfied.

\begin{figure}
\centering
\includegraphics[scale=1]{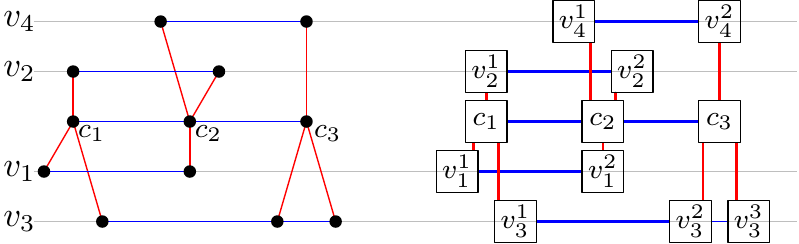}
\caption{
Left: $G_\V$ (red) and $G_\H$ (blue) for Monotone
Not-All-Equal 3SAT instance $\varphi = (v_1 \vee v_2 \vee v_3)(v_4
\vee v_1 \vee v_2)(v_3 \vee v_4 \vee v_3)$.
Right: A USSVR of $\la G_\V , G_\H \ra$ encoding
the truth assignment $ v_2, v_4 = T $ and $ v_1, v_3 = F $.}%
\label{fi:ssvr}
\end{figure}

\paragraph{Construction.}
Let $\varphi$ be an instance of Monotone Not-All-Equal 3SAT with a
set $\mathcal{C}$ of $m$
clauses and a set $\mathcal{V}$ of $n \leq 3m$ variables.
\add{We construct a pair of graphs $ \langle G_\V, G_\H \rangle $
such that each clause and literal in $ \varphi $ is represented by a vertex.}
All clauses form a path in $G_\H$ in the order of their appearance in
$\varphi$; creating one \emph{clause consistency gadget} $G_\H(\mathcal{C})$.
The same holds for all occurrences of literals representing the same
variable; creating $n$ \emph{variable consistency gadgets} $G_\H(v)$
for $v \in \mathcal{V}$.
All occurrences of literals in a clause form a $K_{1,3}$ in $G_\V$, where the clause
vertex is the central vertex; creating $m$ \emph{satisfiability
  gadgets} $G_\V(c)$ for $c \in \mathcal{C}$.
See Fig.~\ref{fi:ssvr} for an example.

Intuitively, the satisfiability gadgets allow us to encode local constraints on the literals for each clause.
We use this to enforce ``not-all-equal" satisfiability.
By contrast, the consistency gadgets allow us to encode global constraints that span multiple clauses; i.e., relating literals that correspond to the same variable.
This completes our construction of $ \langle G_\V, G_\H \rangle $.

\paragraph{Correctness.}
Lemmas~\ref{le:ssvr-implies-satisfiability} and~\ref{le:satisfiability-implies-ssvr} establish the correctness of our reduction.
Hence, since our construction of $ \langle G_\V, G_\H \rangle $ requires $ \Theta(m) $ time, USSVR recognition is \NP-hard.
Note that every USSVR can be redrawn on a $ O(n) \times O(n) $ grid such that its visibilities are unchanged by preserving the order of the endpoints in the $ x $ and $ y $-projections of its unit squares~\cite{casel17}.
This gives a certificate using polynomially-many bits that can be verified in polynomial time.
Thus, USSVR recognition is \NP-complete.

\begin{lemma}\label{le:ssvr-implies-satisfiability}
If $ \langle G_\V, G_\H \rangle $ has a USSVR, $ \varphi $ is satisfiable.
\end{lemma}
\begin{proof}
Let $ \Gamma $ be a USSVR for $ \langle G_\V, G_\H \rangle $.
We construct a truth assignment $ \alpha \colon \mathcal{V} \rightarrow \{T, F\} $ as follows. For every variable $ v \in \mathcal{V} $, we define
\begin{align*}
    \alpha(v) &= \begin{cases}
        T & \text{if $ Y_\Gamma(G_\H(v)) \geq Y_\Gamma(G_\H(\mathcal{C})),$} \\
        F & \text{otherwise}.
    \end{cases}
\end{align*}

We claim that $ \alpha $ satisfies $ \varphi $.
To see this, let us consider any clause vertex $c$ for a clause
$(\ell_1 \vee \ell_2 \vee \ell_3)$.
Note that we distinguish duplicate literals by their order in $\varphi$ as in Fig.~\ref{fi:ssvr}.
By construction, $ c $ is a cut vertex whose removal from $G_\V$
creates three components, each containing one literal
vertex $\ell_i$.
Then by Lemma~\ref{le:Nestedness}, for at least one such vertex, say $
\ell_2 $, $ X_\Gamma(\ell_2) \subseteq X_\Gamma(c) $.
But in fact, since $ \Gamma $ is a unit-square representation, we have
$ X_\Gamma(\ell_2) = X_\Gamma(c) $.
Hence, for every $ k \in \{1, 3\} $, $ X_\Gamma(\ell_k) \cap X_\Gamma(\ell_2) \neq \emptyset $.
Applying Property~\ref{prop:nextBarsSameSide}, we see that $ \Gamma(\ell_2) $ and $ \Gamma(\ell_k) $ must not be both above or both below $ \Gamma(c) $.
Moreover, since every consistency gadget in our construction is a component of $ G_\H $, Lemma~\ref{le:NoTwist} implies that for all variables $ v \in \mathcal{V} $, either $ Y_\Gamma(G_\H(v)) \geq Y_\Gamma(G_\H(\mathcal{C})) $
\add{or $ Y_\Gamma(G_\H(\mathcal{C})) \geq Y_\Gamma(G_\H(v)) $.}
Therefore,  either $ Y_\Gamma(\ell_2) \geq Y_\Gamma(c) \geq Y_\Gamma(\ell_k) $ for $k \in \{1,3\}$, implying that $ \alpha $ satisfies exactly one literal in $ c $, or $ Y_\Gamma(\ell_k) \geq Y_\Gamma(c) \geq Y_\Gamma(\ell_2) $, implying that $ \alpha $ satisfies exactly two.
By repeating this argument for all clauses in $ \mathcal{C} $, we see that $ \alpha $ satisfies $ \varphi $.
\end{proof}

\begin{lemma}\label{le:satisfiability-implies-ssvr}
If $ \varphi $ is satisfiable, $ \langle G_\V, G_\H \rangle $ has a USSVR.
\end{lemma}
\begin{proof}
Let $ \alpha \colon \mathcal{V} \rightarrow \{T, F\} $ be a truth assignment satisfying $ \varphi $. 
To construct a USSVR $ \Gamma $ for $ \langle G_\V, G_\H \rangle $, we first represent $ G_\V $ and $ G_\H $ as two sets of intervals on the $ x $ and $ y $-axes respectively.
The construction of these intervals is as follows.

For the $i$th clause $ c = ( \ell_1 \vee \ell_2 \vee \ell_3) \in
\mathcal{C} $, since $ \alpha $ satisfies 
exactly one or two of its literals, there must be one, say $ \ell_2 $,
that has a unique truth value.
We represent both $ c $ and $ \ell_2 $ on the $ x $-axis by the interval $ [3i + 1, 3i + 2] $.
Moreover, assuming $ \ell_1 $ and $ \ell_3 $ are in order, we
represent their corresponding literal vertices on the $ x $-axis as
the intervals $ [3i, 3i + 1] + \epsilon $ and $ [3i + 2, 3i + 3] -
\epsilon $ respectively, for some small $ \epsilon > 0 $.

Let $ \rho \colon \mathcal{V} \cup \{\mathcal{C}\} \rightarrow \{0, \ldots, |\mathcal{V}|\} $ be a bijection satisfying $ \rho(v) > \rho(\mathcal{C}) $ if and only if $ \alpha(v) = T $ for each $ v \in \mathcal{V} $.
For each variable consistency gadget $ G_\H(v) $, we represent its vertices on the $ y $-axis by the interval $ [2\rho(v), 2\rho(v) + 1] $.
We also represent the clause consistency gadget similarly, replacing $ \rho(v) $ with $ \rho(\mathcal{C}) $.

Observe that each vertex in $ V $ is represented by two unit intervals, one on the $ x $-axis and one on the $ y $-axis.
Thus, for each $ u \in V $, we can define $ \Gamma(u) $ to be the Cartesian product of its two corresponding intervals.
To see that this gives a valid USSVR $ \Gamma $ for $ \langle G_\V, G_\H \rangle $, we make three observations.
\begin{enumerate}
\item Every gadget in $ G_\V $ (resp., $ G_\H) $ occupies a
  contiguous interval on the $ x $-axis (resp., $ y $-axis)
  that is disjoint from the intervals of other gadgets.
\item Every satisfiability gadget in $ G_\V $ for a clause $ c = (
  \ell_1 \vee \ell_2 \vee \ell_3)$ is drawn such that $ \Gamma(c)
  $ blocks vertical visibility between $ \Gamma(\ell_2) $ and $
  \Gamma(\ell_k) $ for $ k \in \{1, 3\} $ assuming $\ell_2$ has
  the unique unique truth value of literals in $c$.
\item Every consistency gadget in $ G_\H $ is drawn as a horizontal stack of unit squares (in order from left to right) that share a $ y $-projection.
\end{enumerate}

The first observation implies that no two gadgets in $ G_\V $ (resp., $ G_\H $) share an (unwanted) visibility.
The next two observations mean that the implied visibilities for each gadget in $ G_\V $ and $ G_\H $ are realized exactly.
\end{proof}

\subsection{Hardness of RSVR recognition}

\label{sec:rsvr-recognition}
In this section, we prove that determining if a pair of undirected graphs has a \emph{simultaneous visibility representation using rectangles} (RSVR) is \NP-complete.
In contrast to the proof given in Section~\ref{sec:ssvr-recognition}, here, we reduce from the \NP-complete problem of 3SAT~\cite{cook71}.
A new reduction is needed since 
every pair of edge-disjoint caterpillar forests (as produced for the reduction in Section~\ref{sec:ssvr-recognition}) has an RSVR due to Theorem 5 by Bose et al.~\cite{bose96}.

Our modified construction
is not much more complicated than before:
one graph remains a set of disjoint paths while the other is a set of
disjoint trees with 3 leaves connected to a root by paths of length 2.
This slight modification allows us to prove the following theorem.

\begin{theorem}\label{thm:rsvr-is-hard}
    Deciding if a pair of undirected graphs has an RSVR is \NP-complete.
\end{theorem}

\begin{figure*}
\centering
\includegraphics[scale=0.65]{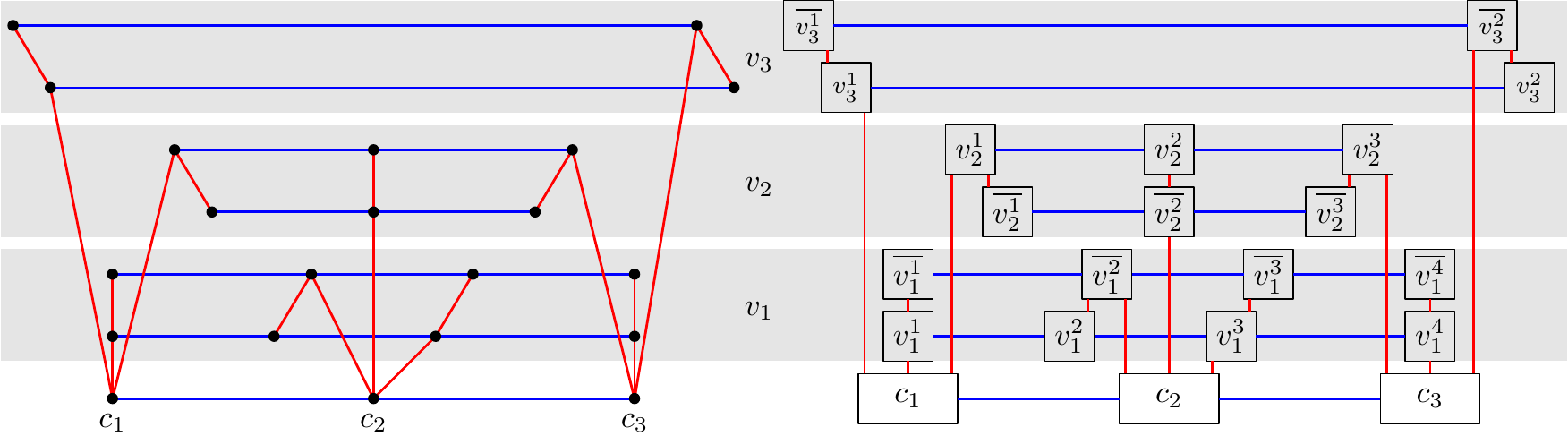}
\caption{
Left: $G_\V$ (red) and $G_\H$ (blue) for 3SAT instance
$\varphi = (v_3 \vee v_1 \vee v_2) (\overline{v_1} \vee \overline{v_2} \vee v_1) (v_2 \vee v_1 \vee \overline{v_3}) $.
Right: An RSVR of $\la G_\V , G_\H \ra$ encoding
the truth assignment $ v_1, v_3 = T $ and $ v_2 = F $.}
\label{fi:rsvr}
\end{figure*}

\textbf{Construction.}
Let $ \varphi $ be an instance of 3SAT with a set $ \mathcal{C} $ of $m$ clauses and a set $ \mathcal{V} $ of $n \leq 3m$ variables.

We adapt the gadgets used in Section~\ref{sec:ssvr-recognition} to this setting as follows.
Each satisfiability gadget $G_\V(c)$ for $ c \in \mathcal{C} $ is now a 1-subdivision of $K_{1,3}$ (i.e. $\spider$) where the central vertex is the clause $c$, the subdivision vertices are the occurrences of literals in the clause, and each leaf is an occurrence of the negation of its parent.
In addition to the variable consistency gadget, we also construct a \emph{negated variable consistency gadget} $ G_\H(\overline{v}) $ for each variable $v \in \mathcal{V}$ that is the path of negated occurrences of literals in the order of their appearance in $\varphi$.
This completes our construction of $ \la G_\V, G_\H \ra $; see Fig.~\ref{fi:rsvr} for an example. 

\textbf{Correctness.}
Lemmas~\ref{le:rsvr-implies-satisfiability} and \ref{le:satisfiability-implies-rsvr} establish the correctness of our reduction.
Thus, by a similar argument to the one found in Section~\ref{sec:ssvr-recognition}, RSVR recognition is \NP-complete.

\begin{lemma}\label{le:rsvr-implies-satisfiability}
	If $ \langle G_\V, G_\H \rangle $ has an RSVR, $ \varphi $ is satisfiable.
\end{lemma}
\begin{proof}
    Let $ \Gamma $ be an RSVR for $ \langle G_\V, G_\H \rangle $.
    If, for some variable $ v \in \mathcal{V} $, we have $ Y_\Gamma(G_\H(\overline{v})) \geq Y_\Gamma(G_\H(v)) \geq Y_\Gamma(G_\H(\mathcal{C})) $
    \add{or $ Y_\Gamma(G_\H(\mathcal{C})) \geq Y_\Gamma(G_\H(v)) \geq Y_\Gamma(G_\H(\overline{v})) $},
    we say that $ v $ is \emph{positively-arranged} in $ \Gamma $.
    We construct a truth assignment $ \alpha \colon \mathcal{V} \rightarrow \{T, F\} $ as follows. For each variable $ v \in \mathcal{V} $, we 
    define
    \begin{align*}
        \alpha(v) &= \begin{cases}
            T & \text{if $ v $ is positively-arranged in $ \Gamma, $} \\
            F & \text{otherwise}.
        \end{cases}
    \end{align*}
    
    We claim that $ \alpha $ satisfies $ \varphi $.
    To see this, let us consider any clause vertex $c$ for a clause $(\ell_1 \vee \ell_2 \vee \ell_3)$.
    By construction, $ c $ is a cut vertex whose removal from $ G_\V(c) $ creates three components, each containing one literal vertex and its negated counterpart.
    Then by Lemma~\ref{le:Nestedness}, for at least one literal, say $ \ell_2 $, we have $ X_\Gamma(\{\ell_2,\overline{\ell_2}\}) \subseteq X_\Gamma(c) $.
    Hence, $ X_\Gamma(\overline{\ell_2}) \cap X_\Gamma(c) \neq \emptyset $.
    
    Applying Property~\ref{prop:nextBarsSameSide}, we see that $ \Gamma(\overline{\ell_2}) $ and $ \Gamma(c) $ must not be both above or both below $ \Gamma(\ell_2) $.
    Moreover, since the consistency gadgets in our construction are components of $ G_\H $, Lemma~\ref{le:NoTwist} implies that their $ y $-projections must form disjoint intervals.
    Therefore, either $ Y_\Gamma(\overline{\ell_2}) \geq Y_\Gamma(\ell_2) \geq Y_\Gamma(c) $ or \add{$ Y_\Gamma(c) \geq Y_\Gamma(\ell_2) \geq  Y_\Gamma(\overline{\ell_2}) $}.
    Thus, if $ \ell_2 $ is a positive literal then its variable is positively-arranged in $ \Gamma $; otherwise, $ \ell_2 $ is a negative literal implying that its variable is not positively-arranged in $ \Gamma $.
    In either case, $ \alpha $ satisfies $ c $.
    By repeating this argument for all clauses in $ \mathcal{C} $, we see that $ \alpha $ satisfies $ \varphi $.
\end{proof}

\begin{lemma}\label{le:satisfiability-implies-rsvr}
	If $ \varphi $ is satisfiable, $ \langle G_\V, G_\H \rangle $ has an RSVR.
\end{lemma}
\begin{proof}
    Let $ \alpha \colon \mathcal{V} \rightarrow \{T, F\} $ be a truth assignment satisfying $ \varphi $. 
    To construct an RSVR $ \Gamma $ for $ \langle G_\V, G_\H \rangle $, we first represent $ G_\V $ and $ G_\H $ as two sets of intervals on the $ x $ and $ y $-axes respectively.
    The construction of these intervals is as follows.
    
    For the $i$th clause $c = (\ell_1 \vee \ell_2 \vee \ell_3) \in \mathcal{C}$, we represent $c$ on the $ x $-axis by the interval $ [7i + 2, 7i + 5] $.
    Next, for one of the satisfied literals in $c$, say $ \ell_2 $, we represent both $ \ell_2 $ and $ \overline{\ell_2} $ on the $ x $-axis by the interval $ [7i + 3, 7i + 4] $.
    Finally, assuming that $ \ell_1 $ and $ \ell_3 $ are in order, we represent $ \ell_1 $ and $ \overline{\ell_1} $ on the $ x $-axis by the intervals $ [7i + 1, 7i + 2] + \epsilon $ and $ [7i, 7i + 1] + 2\epsilon $ respectively, for some positive but small $ \epsilon $. 
    Similarly, we represent $ \ell_3 $ and $ \overline{\ell_3} $ by the intervals $ [7i + 5, 7i + 6] - \epsilon $ and $ [7i + 6, 7i + 7] - 2\epsilon $ respectively.
    
    For the $j$th variable $ v \in \mathcal{V} $, if $ \alpha(v) = T $, we represent the vertices in $ G_\H(v) $ and $ G_\H(\overline{v}) $ on the $ y $-axis by the intervals $ [4j, 4j + 1] $ and $ [4j + 2, 4j + 3] $.
    Otherwise, if $ \alpha(v) = F $, we simply swap the intervals and proceed as before.
    Finally, we represent every vertex in $ G_\H(\mathcal{C}) $ on the $ y $-axis by the interval $ [0, 1] $.
    
    Observe that each vertex in $ V $ is represented by two (nonempty) intervals, one on the $ x $-axis and one on the $ y $-axis.
    Thus, for every $ u \in V $, we can define $ \Gamma(u) $ to be the Cartesian product of its two corresponding intervals.
    To see this gives a valid RSVR $ \Gamma $ for $ \langle G_\V, G_\H \rangle $, we make three observations.
    \begin{enumerate}
    	\item Every gadget in $ G_\V $ (resp., $ G_\H $) occupies a contiguous interval on the $ x $-axis (resp., $ y $-axis) that is disjoint from the intervals of other gadgets.
        \item Every satisfiability gadget in $ G_\V $ for a clause $ c = (\ell_1 \vee \ell_2 \vee \ell_3) $ is drawn such that $ \Gamma(\ell_2) $ blocks vertical visibility between $ \Gamma(c) $ and $ \Gamma(\overline{\ell_2}) $. Moreover, $ X_\Gamma(c) $ intersects $ X_\Gamma(\ell_k) $ but not $ X_\Gamma(\overline{\ell_k}) $ for $ k \in \{1, 3\} $.
        \item Every consistency gadget in $ G_\H $ is drawn as a horizontal stack of rectangles (in order from left to right) that share a $ y $-projection.
    \end{enumerate}
    
    The first observation implies that no two gadgets in $ G_\V $ (resp., $ G_\H $) share an (unwanted) visibility. The next two observations mean that the implied visibilities for each gadget in $ G_\V $ and $ G_\H $ are realized exactly. Therefore, $ \Gamma $ is indeed a valid RSVR for $ \langle G_\V, G_\H \rangle $.
\end{proof}

\subsection{Generalizations}
\label{sec:hardness-generalization}

Notice that in the reduction given in Section~\ref{sec:ssvr-recognition}, we make only the assumption that the $ x $-projection of every allowable shape has the same size.
Thus, we can adapt this reduction to any family
of shapes that share a fixed positive width.
Moreover, the reduction in Section~\ref{sec:rsvr-recognition} can be adapted to any family
of shapes for which at least two have different widths; e.g., the family of L-shapes.
These observations allow us to state the following.

\begin{corollary}
Deciding if a pair of undirected graphs has an SVR using shapes from a family
of connected shapes with positive width and height is \NP-hard.
\end{corollary}

For families
of orthogonal polygonal paths with constant complexity
(e.g., L-shapes), SVR recognition is also in \NP; this follows by a similar argument to what we gave for USSVR recognition.

\begin{corollary}
Deciding if a pair of undirected graphs has an SVR using shapes from a family
of orthogonal polygonal paths with constant complexity and positive width and height is \NP-complete.
\end{corollary}

\section{Pairs of undirected paths}
\label{sec:paths}

The hardness results of Section~\ref{sec:np-complete} use graphs that are
not significantly more structurally complicated than paths.
This motivates the question of whether pairs of (undirected) paths always admit SVRs and if not, whether there exists a polynomial time algorithm to decide when they do.

This question has an easy answer when the underlying shapes are rectangles:
Two paths $\la \pv=(V,E_\V),\ph=(V,E_\H) \ra$ defined on the same vertex set have an RSVR (in fact, a USSVR) if and only if $E_\V\cap E_\H\neq\emptyset$.
If the paths share an edge then both $x$ and $y$-projections of two rectangles must overlap implying that the shapes themselves overlap.
If the paths do not share an edge then the following algorithm creates a USSVR:

\begin{quote}
Algorithm \textsf{A}:
For all $v\in V$, place $\G(v)$ in the plane with its bottom-left corner at $(i,j)$ where $i$ (resp., $j$) is $v$'s place along $\pv$ (resp., $\ph$) from a fixed reference endpoint of the path; and set the side lengths of the rectangles to be $1+\eps$ for a small $\eps>0$.
\end{quote}

Brass et al.~\cite{brass07} presented this algorithm to simultaneously embed, without self-intersection, two paths using points as vertices and line segments as edges.
In even earlier work, Bose et al.~\cite{bose96} use a similar approach to obtain rectangle visibility graphs for the union of two edge-disjoint caterpillar forests on the same vertex set (the Caterpillar Theorem).
Their approach implies an RSVR for the two edge-disjoint caterpillars, which includes edge-disjoint paths, and, as they point out, extends to some other graph classes.
\comment{
\begin{figure}[t]
\centering
\begin{subfigure}[t]{.45\textwidth}
\centering
\includegraphics{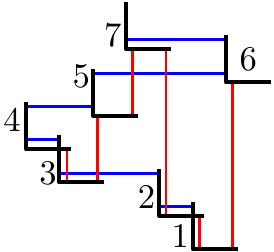}
\end{subfigure}
\hfill
\begin{subfigure}[t]{.45\textwidth}
\centering
\includegraphics{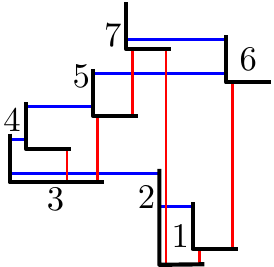}
\end{subfigure}
\caption{Left: The result of running Algorithm \textsf{A} on $\ph=(1,2,3,4,5,6,7)$ and $\pv=(4,3,5,7,2,1,6)$. Note the intersection of $\Gamma(4),\Gamma(3)$ and $\Gamma(2),\Gamma(1)$. Right: The intersections can be alleviated by stretching $\Gamma(3)$ and $\Gamma(2)$. 
}
\label{fig:lsvr_example}
\end{figure}
}

It is easy to check that the rectangles placed by Algorithm \textsf{A} satisfy all the required visibilities. Thus, there exists an RSVR and USSVR of $\pvph$ if and only if $E_\V\cap E_\H=\emptyset$.
In fact, the result holds for any shapes that intersect if both of their $x$ and $y$-projections overlap. We turn our attention, therefore, to shapes that do not obey this property. Surprisingly, this question becomes significantly more complicated 
even for L-shapes which are simply the left and bottom sides of a rectangle.
A \emph{simultaneous visibility representation using fixed orientation L-shapes $ \{\llcorner\} $} (an LSVR) of $\gvgh$ is a pair $\la \Gv,\Gh\ra$ where $\Gv$ is a BVR, $\Gh$ is a BVR rotated by $\pi/2$ (with vertical bars and horizontal visibility), and for all $v\in V$, $\bv(v)=\lh(v)$ (i.e., $\Gv(v)$ and $\Gh(v)$ share their respective bottom and left endpoints).

\subsection{LSVR of two undirected paths}
\label{sec:lsvrTwoPaths} 

Let $\pvph$ be two undirected paths defined on the same set of vertices $V=[n]$.
By relabeling the vertices, we may assume that $P_\H$ is the path $(1,2,\dots,n)$ and $P_\V$ is $(\pi_1, \pi_2, \dots, \pi_n)$ for a permutation $\pi$ of $[n]$.
While the paths are undirected, the algorithm can direct each path in
one of two ways, from left-to-right or from right-to-left, which
corresponds to having the reference endpoint (used by Algorithm
\textsf{A}) at the beginning or end of the path.

\begin{figure}[t]
\centering
\begin{subfigure}[t]{.45\textwidth}
\centering
\includegraphics[page=2]{anna_box.pdf}
\caption{For $\pv=(2,1,6,4,5,3)$ (and $\ph=(1,2,3,4,5,6)$)
shapes $\Gamma(4),\Gamma(5)$ and
$\Gamma(1),\Gamma(2)$ intersect in some drawings.
Stretching $\Gamma(1)$ west (or $\Gamma(2)$ south) in the \SW{} drawing removes the $\Gamma(1),\Gamma(2)$ intersection without introducing additional visibilities.}
\label{fig:216453}
\end{subfigure}
\hfill
\begin{subfigure}[t]{.45\textwidth}
\centering
\includegraphics[page=1]{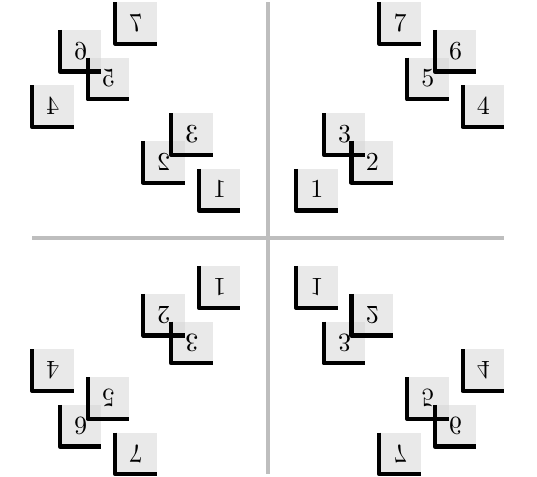}
\caption{For $\pv=(1,3,2,7,5,6,4)$ (and $\ph=(1,2,3,4,5,6,7)$) no stretching in any drawing avoids additional visibilities.  This pair of paths does not have an LSVR.}
\label{fig:1327564}
\end{subfigure}
\caption{Two examples of the four drawings from Algorithm \textsf{A}.}
\end{figure}

\begin{figure}
\centering
\begin{subfigure}[t]{.3\textwidth}
\centering
\includegraphics[page=1]{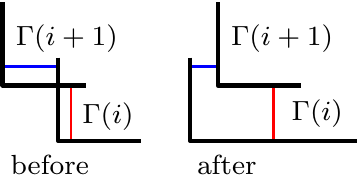}
\caption{}
\label{fig:trans1}
\end{subfigure}
\hfill
\begin{subfigure}[t]{.3\textwidth}
\centering
\includegraphics[page=2]{stretch.pdf}
\caption{}
\label{fig:trans2}
\end{subfigure}
\hfill
\begin{subfigure}[t]{.35\textwidth}
\centering
\includegraphics[page=3]{stretch.pdf}
\caption{}
\label{fig:stretch}
\end{subfigure}
\caption{Two ways to stretch one L in a pair of intersecting L-shapes in a \SW{} drawing to remove a crossing. To perform (a), no other shape can have a visibility with $\G(i+1)$ from the south. To perform (b), no other shape can have a visibility with $\G(i)$ from the west.
(c) For $\pv=(5,3,2,4,1)$, multiple shapes stretch to remove one crossing and to avoid a new vertical visibility. 
}
\label{fig:transformations}
\end{figure}

We construct an LSVR for $\pvph$ as follows:
(1) Run Algorithm \textsf{A} using L-shapes for all four choices of
the two paths' orientations.
These correspond to using the South+West (\SW), South+East (\SE), North+West (\NW), and
North+East (\NE) sides of the $(1+\epsilon)\times(1+\epsilon)$ (possibly
overlapping) squares drawn
by Algorithm \textsf{A} when it attempts to draw a USSVR of the two paths and then reflecting the drawing vertically or horizontally, see Fig.~\ref{fig:216453}.
We assume that $P_\H=(1,2,\dots,n)$ and the reference point for each
path is its leftmost vertex in the \SW{} drawing, so $\G(i)$ and
$\G(i+1)$ intersect if and only if 
$(i+1,i)\in \pv$ for the \SW{} and \NE{} drawings
or
$(i,i+1)\in \pv$ for the \SE{} and \NW{} drawings.
For a given pair of paths, all or two or none of these drawings
contain intersecting L's.
(2) If some drawing does not contain intersections, we're done.
Otherwise, ``stretch'' some of the L's in some drawing so that each pair of
intersecting L's ``nest'' while preserving their existing
visibilities, see Fig.~\ref{fig:transformations}.
This is not always possible, see Fig.~\ref{fig:1327564}, or it may be possible in only some of the four drawings, see Fig.~\ref{fig:216453}.
We prove that an LSVR for $\pvph$ exists if and only if this stretching transformation is possible for some drawing
(Theorem~\ref{thm:lsvrtwopaths}) by describing a drawing procedure (Section~\ref{sec:Sufficiency}) and showing that if this procedure is
unable to produce an LSVR then no LSVR for $\pvph$ exists (Section~\ref{sec:Necessity}).

Given, for example, a \SW{} drawing, we would like to know what intersections of L's can be removed by stretching.
If $\pi_1 > \pi_2 > \dots > \pi_c$ 
then the shapes for vertices $\pi_1, \pi_2, \dots, \pi_c$
are not only the $c$ west-most (i.e., leftmost) shapes in increasing $x$-order,
they are also in decreasing $y$-order.
By stretching every one of these shapes westward to reverse the $x$-order of their vertical bars, we can eliminate any intersection of these shapes while maintaining exactly the same vertical (and, of course, horizontal) visibilities, see Fig.~\ref{fig:stretch}.
A similar transformation can reverse the $y$-order of the south-most shapes to remove intersections under symmetric conditions.

We first describe the sequences of reversible vertices and then state the conditions that allow a drawing to be stretched to remove intersections.
We say a sequence $S=(s_1, s_2, \dots, s_k)$ is \emph{increasing} (\emph{decreasing}) in a sequence $T=(t_1, t_2, \dots, t_n)$ if
there exist indices $(j_1, j_2, \dots, j_k)$ that are strictly increasing (decreasing) such that $s_i = t_{j_i}$ for all $i \in [k]$.
A sequence $S$ is \emph{monotonic} in $T$ if it is either increasing or decreasing in $T$.
For example, $(4,7,3)$ is monotonic in $(1,3,2,7,5,6,4)$ but $(3,4,7)$ is not.

\begin{definition}
\label{def:crossing_parts}
For $P_\H = (1,2, \dots, n)$ and $P_\V = (\pi_1, \pi_2, \dots, \pi_n)$ where $\pi$ is a permutation of $[n]$, let 
\begin{enumerate}
\item $\phprim = (1,2,\dots,a)$ be the longest such sequence monotonic in $\pi$,
\item $\pvprim = (\pi_1,\pi_2,\dots,\pi_c)$ be the longest such sequence monotonic in $[n]$, and 
\end{enumerate}
\end{definition}
For example, if $P_\V = (1,3,2,7,5,6,4)$, then $\phprim = ( 1,2 )$
and
$\pvprim = (1,3)$
. 
The \SW{} drawing can be stretched to remove all intersections if
\begin{equation}
\label{sw}
\text{for all } (i+1,i)\in P_\V, i\in \pvprim \text{ or } i+1\in \phprim.
\end{equation}

By changing the reference endpoints of the two paths, we obtain the
\NW{}, \SE{}, and \NE{} drawings.  Each of these can be stretched to
remove all intersections if condition~\eqref{sw} holds, after
renumbering the vertices so that the (possibly re-oriented) path
$P_\H$ is $(1,2,\dots,n)$.


The rest of Section \ref{sec:lsvrTwoPaths} will focus on proving the following theorem.

\begin{theorem}
\label{thm:lsvrtwopaths}
Let $\pvph$ be two paths defined on the same set of $n$
vertices. There exists an LSVR of $\pvph$ if and only if
condition~\eqref{sw} holds for at least one drawing \SW, \NW, \SE, \NE.
In the positive case, the LSVR is realizable in $O(n)$ time on a grid of size $O(n)\times O(n)$. 
\end{theorem}

Note that since condition~\eqref{sw} can be tested in linear time,
Theorem~\ref{thm:lsvrtwopaths} yields a linear time algorithm to
determine if two given paths admit an LSVR.

\subsubsection{Condition~\eqref{sw} implies LSVR}
\label{sec:Sufficiency}
In this section we present an algorithm which constructs an LSVR for a
pair of paths assuming condition~\eqref{sw} holds for some
drawing.
We describe the algorithm for the \SW{} drawing.
The other drawings are checked in a similar fashion after 
changing the reference endpoints and renumbering the vertices so that
the (possibly re-oriented) path $P_\H$ is $(1,2,\dots,n)$.

\begin{figure}
\centering
\begin{subfigure}[t]{.2\textwidth}
\centering
\includegraphics[page=1,width=\linewidth]{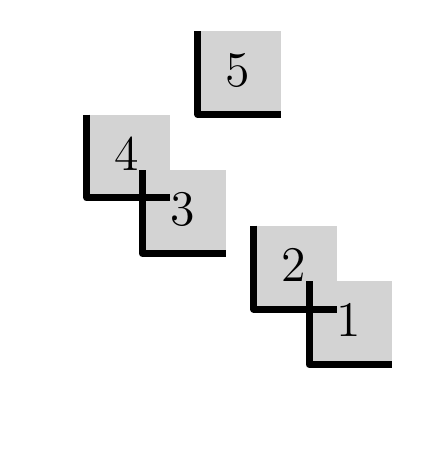}
\caption{Step 1}
\label{fig:lsvr_paths1}
\end{subfigure}
\hfill
\begin{subfigure}[t]{.2\textwidth}
\centering
\includegraphics[page=2,width=\linewidth]{lsvr_paths.pdf}
\caption{Step 2}
\label{fig:lsvr_paths2}
\end{subfigure}
\hfill
\begin{subfigure}[t]{.2\textwidth}
\centering
\includegraphics[page=3,width=\linewidth]{lsvr_paths.pdf}
\caption{Step 3}
\label{fig:lsvr_paths3}
\end{subfigure}
\caption{Running \lsvrpaths on $ \pv = (4, 3, 5, 2, 1) $. In Step 1, we run Algorithm \textsf{A} to produce a \SW{} drawing. In Step 2, we stretch $\G(\pi_i)$ to the left for $ \pi_i \in \pvprim = \{4, 3\} $. In Step 3, we stretch $ \G(i) $ downwards for $ i \in \phprim \setminus \pvprim = \{1, 2\} $.}
\label{fig:lsvr_paths}
\end{figure}

\paragraph{\lsvrpaths Algorithm.} Let $\la \pv=(\pi_1,\pi_2, \dots ,\pi_n), \ph=(1,2,\dots,n) \ra$ be two paths defined on the same vertex set $[n]$.
We break the algorithm into three steps. 

\textbf{Step 1}: Run Algorithm \textsf{A} to produce a \SW{} drawing (with possible intersections).

\textbf{Step 2}: Let $\pvprim=\{\pi_1,\dots,\pi_c\}$. If $\pi_1>\pi_2>\dots>\pi_c$ and there are crossings in $\G(\pvprim)$, then for all $\pi_i\in\pvprim$, stretch $\G(\pi_i)$ to the left such that $\l(\pi_i)=2-i$.

\textbf{Step 3}: Let $\phprim=\{1,\dots,a\}$. If
$(1,2,\dots,a)$ is decreasing in $\pi$
and there are crossings in $\G(\phprim)$, then for all $i\in\phprim\setminus\pvprim$, stretch $\G(i)$ downwards such that $\b(i)=2-i$.

See Fig.~\ref{fig:lsvr_paths} for an illustration of the algorithm.
Note  that we could have stretched $\phprim$ in step 2 and $\pvprim\setminus\phprim$ in step 3.
Observe that \lsvrpaths requires linear time. Furthermore, the layout
is contained in $[2-n,n]\times[2-n,n]$, i.e., a grid of size
$O(n)\times O(n)$. Hence, the following lemma proves one direction of
Theorem \ref{thm:lsvrtwopaths}.

\begin{lemma}
\label{lem:algCorrect}
If $\pvph$ are two paths defined on the same vertex set and
condition~\eqref{sw} is true then Algorithm \lsvrpaths returns an LSVR
of $\pvph$.
\end{lemma}

We break the proof into a sequence of lemmas. 

\begin{lemma}
\label{prop:reversed}
Let $\phprim=\{1,\dots,a\}$ and $\pvprim=\{\pi_1,\dots,\pi_c\}$. 
After Step 1 there is a crossing among $\G(\phprim)$ only if $(1,2,\dots,a)$ is decreasing in $\pi$ and among $\G(\pvprim)$ only if $\pi_1>\pi_2>\dots>\pi_c$. 
\end{lemma}
\begin{proof}
If $(1,2,\dots,a)$ is increasing in $\pi$, then after Step 1 we have $\l(\pi_i)<\l(\pi_{i+1})$ and $\b(i)<\b(i+1)$ for all $i<a$, hence there are no crossings in $\G(\phprim)$. The argument for $\pvprim$ is similar. 
\end{proof}

\begin{lemma}
\label{prop:AwithoutC}
Let $\phprim=\{1,\dots,a\}$. If
$(1,2,\dots,a)$ is decreasing in $\pi$,
then either $\phprim\setminus\pvprim=\emptyset$ or $\phprim\setminus\pvprim=\{1,\dots,a'\}$ for some $a'\leq a$. 
\end{lemma}
\begin{proof}
Let $i\in \phprim\cap\pvprim$ (if no such $i$ exists, we are done).
Suppose $i+1\in\phprim$.
Since $(1,2,\dots,a)$ is decreasing in $\pi$, $i+1$ must precede $i$ in $\pi$.
Since $\pvprim$ is a prefix of $\pi$ that contains $i$, we have $i+1 \in \pvprim$.
\end{proof}

\begin{lemma}
\label{lem:visibilitiesCorrectAlg}
At the end of \lsvrpaths, the required visibilities 
are present in $\G$ and no others. 
\end{lemma}

\begin{proof}
We first prove the claim for horizontal visibilities. By Lemma \ref{prop:AwithoutC}, write $\phprim\setminus\pvprim=\{1,\dots,a'\}$. Consider the vertices $\{a',\dots,n\}$, whose vertical bars are not altered after Step 1. By construction, for all $i\in\{a'+1,\dots,n\}$, $\projy(i)=[i,i+1+\eps]$. Therefore, $\projy(i)\cap\projy(i+1)=[i+1,i+1+\eps]$ and $\projy(j)\cap[i+1,i+1+\eps]=\emptyset$ for all $j\neq i,i+1$ (even if Step 3 is performed). Hence, $\Gh(i)$ and $\Gh(i+1)$ share a horizontal visibility for all $i\in\{a',\dots,n-1\}$. Moreover, $\projy(i)\cap\projy(j)=\emptyset$ for $j\notin\{i-1,i,i+1\}$, so there are no unwanted visibilities among these shapes. If Step 3 is not performed, then the same argument demonstrates that all required horizontal visibilities are present among $[n]$.
Suppose, therefore, that Step 3 is performed. Fix $i\in\{2,\dots,a'\}$. To complete the proof, it suffices to show that $\Gh(i)$ shares a visibility with $\Gh(i-1)$ and does not share a visibility with $\Gh(j)$ for $j\leq i-2$. By construction, $\projy(i)=[2-i,i+1+\eps]$ for all $i\in[a']$. Hence, $\projy(i)\supset\projy(i-1)$. Furthermore, since Step 3 was performed, $(1,2,\dots,a')$ is decreasing in $\pi$ and so $\l(1)>\l(2)>\dots>\l(a')$ demonstrating that $\Gh(j)$ for $j\leq a'$ cannot block the visibility between $\Gh(i)$ and $\Gh(i-1)$. Moreover, $\b(j)\geq a'+1\geq i+1$ for all $j\geq  a'+1$; hence, neither can $\Gh(j)$ for $j\geq a'+1$. Finally, since $\projy(i-1)\supset\projy(i-2)\supset\dots\supset\projy(1)$, we see that $\projy(i-1)$ blocks $\Gh(i)$ from sharing a visibility with $\Gh(j)$ for $j\leq i-2$. The proof of vertical visibilities is almost identical, except that the argument uses $y$-projections. 
\end{proof}

\begin{lemma}
\label{lem:noCrossingsAlg}
After Algorithm \lsvrpaths is complete, there are no crossings among any shapes. 
\end{lemma}
\begin{proof}
First we observe that after Step 1 of the Algorithm, there is a crossing between two shapes $\G(i_1)$ and $\G(i_2)$ if and only if we can write $i_1=i$ and $i_2=i+2$ and $(i+1,i)\in \pv$. Consequently, we may write $i+1=\pi_j$ and $i=\pi_{j+1}$ for some $j$.  Since condition~\eqref{sw} is true, either $i\in \pvprim$ or $i+1\in\phprim$. First we consider the case when $i\in\pvprim$. After Step 2 is carried out, we have $\l(i)=2-j-1<2-j=\l(i+1)$ and since there was no vertical displacement of the bars, this alleviates the crossing between $\Gh(i)$ and $\Gv(i-1)$. It remains to show that the transformation did not induce any further crossings. Notice that even after Step 3, the only shapes which share an $x$-coordinate with $\G(i)$ are $\G(i-1)$ and $\G(i+1)$. This is because $i>a'$ where $\phprim\setminus\pvprim=[a']$, hence the modifications to $\G(\phprim\setminus\pvprim)$ (if any) stretch the shapes \emph{downwards} and thus $\projy(\phprim\setminus\pvprim)\cap\projy(i)$ is unaffected. Furthermore, $\G(i-1)$ intersects $\G(i)$ iff they share a vertical visibility, in which case we must have $i-1=\pi_{j+2}$ (since $i+1=\pi_j$) and, by Lemma \ref{lem:visibilitiesCorrectAlg}, $\G$ contains only the correct vertical visibilities. In this case, notice that $\pi_{j+2}\in \pvprim$ since $\pi_{j+1}\in\pvprim$ and by Lemma \ref{prop:reversed}, $\pi_1>\pi_2>\dots>\pi_c$ and $\pi_{j+2}=i-1<i=\pi_{j+1}$. This completes the proof if $i\in \pvprim$. If $i+1\in\phprim$ instead, the argument is similar, except we argue about Step 3 and the vertical displacement of $\G(i+1)$. 
\end{proof}

\subsubsection{Orderings of \texorpdfstring{$\Gamma$}{Gamma} and Two Properties of BVRs}

For a path $P=(v_1,\dots,v_n)$ and a BVR $\G$ of $P$, we say $\G$ is \df{monotonically increasing} (resp., decreasing) if $\r(v_j)<\r(v_{j+1})$ (resp.,  $\r(v_j)>\r(v_{j+1})$) for all $j\in[n-1]$.
If $\Gamma$ is monotonically increasing or decreasing we say it is \df{monotone}. 
We say $\G$ is \df{strictly increasing} (resp., \df{decreasing}) if $\G$ is monotonically increasing (resp., decreasing) and $X_\Gamma(v_j) \not\subseteq X_\Gamma(v_{j+1})$ (resp.,$X_\Gamma(v_{j+1}) \not\subseteq X_\Gamma(v_j)$) for all $j\in[n-1]$. 
For a BVR rotated by $\pi/2$ (with vertical bars and horizontal visibility), the same definitions apply with $\t(v)$ replacing $\r(v)$. 
A visibility representation in which no vertical or horizontal line contains the endpoints of two shapes from different vertices is called \emph{noncollinear}. Finally, a BVR which is noncollinear and monotone is called \emph{canonical}, and an LSVR $\G=\la\Gv,\Gh\ra$ is \emph{canonical} if $\Gv$ and $\Gh$ are both canonical.
We apply the same definitions and notation to subdrawings of an LSVR $\G$. That is, for a subset $S\subseteq V$, we say $\G(S)$ is monotone (or strictly increasing, etc) if the conditions are satisfied for the realizations of the vertices in $S$.

Before continuing on to the proof of necessity, we describe two basic results on the structure of BVRs. 
The first, Property~\ref{prop:cycle}, states a condition that implies the existence of a cycle, see Fig.~\ref{fig:cycle}. 

\begin{figure}[t]
\centering
\begin{subfigure}[t]{.45\textwidth}
\centering
\includegraphics{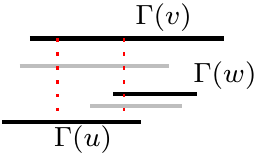}
\caption{Illustration for Property~\ref{prop:cycle}}
\label{fig:cycle}
\end{subfigure}
\hfill
\begin{subfigure}[t]{.45\textwidth}
\centering
\includegraphics{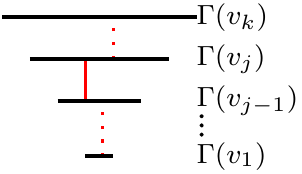}
\caption{Illustration for Property~\ref{prop:tower}}
\label{fig:tower}
\end{subfigure}
\caption{The gray lines in (a) represent other bars which may or may not be in the BVR, and the dotted (red) lines represent the two sequences of vertical visibilities.}
\label{fig:cycleAndTower}
\end{figure}

\begin{restatable}{property}{cycle}
\label{prop:cycle}
Let $\Gamma$ be a BVR of $G$.
If an endpoint of $X_\Gamma(w)$ is strictly contained in $X_\Gamma(u) \cap X_\Gamma(v)$ and $\bo(u) < \bo(w) < \bo(v)$ for
$u,v,w \in V$, then there is a cycle in $G$.
\end{restatable}
\begin{proof}
Since an endpoint of $X_\Gamma(w)$ is strictly contained
in $X_\Gamma(u) \cap X_\Gamma(v)$, there exists $x \in X_\Gamma(w)$
and
$x' \not\in X_\Gamma(w)$ such that $x,x' \in X_\Gamma(u) \cap
X_\Gamma(v)$.
By Property~\ref{prop:ThinOverlap},
there exists a path from $u$ to $v$ (following the vertical line through $x$) that, since $\bo(u) < \bo(w) < \bo(v)$,
contains $w$ and a path from $u$ to $v$ (following the vertical line through
$x'$) that does not contain $w$.  The union of these two
paths contains a cycle.
\end{proof}

The second property is that once the $x$-projection of a bar is contained in that of another this containment propagates, for any representation of a path. Fig.~\ref{fig:tower} provides an example. 

\begin{restatable}{property}{tower}
\label{prop:tower}
Let $\G$ be a noncollinear BVR of a path
$P=v_1,\dots,v_n$.
If $X_\G(v_j) \subset X_\G(v_k)$ for $j < k$, then
(i) $X_\G(v_1) \subset X_\G(v_2) \subset \dots \subset X_\G(v_j)$; and
(ii) $\bo(v_1), \dots , \bo(v_k)$ forms a strictly monotonic sequence.
\end{restatable}
\begin{proof}
By assumption we have that $X_\G(v_j) \subset X_\G(v_k)$, and since $\G$ is noncollinear,
we may assume that $\bo(v_j) < \bo(v_k)$ 
or $\bo(v_j) > \bo(v_k)$. Suppose it is the former; the argument is
symmetric in the other case.
Consider the largest $i<j$ such that
either $X_\G(v_i) \not\subset X_\G(v_{i+1})$
or
$\bo(v_i) > \bo(v_{i+1})$.
If $\bo(v_i) > \bo(v_{i+1})$, then $\G(v_i)$ and $\G(v_k)$ are both
above $\G(v_{i+1})$.
By Property~\ref{prop:nextBarsSameSide}, $X_\G(\{v_1, \dots, v_i\}) \cap
X_\G(\{v_k, \dots, v_n\}) = \emptyset$.
However, since $X_\G(v_i) \cap X_\G(v_{i+1}) \neq \emptyset$ (they
must share a visibility) and
$X_\G(v_{i+1}) \subset  X_\G(v_k)$, it follows that $ X_\G(v_i)\cap
X_\G(v_k) \neq \emptyset$, a contradiction.
Otherwise, if $X_\G(v_i) \not\subset X_\G(v_{i+1})$, then one of the
endpoints of $X_\G(v_{i+1})$ is contained in $X_\G(v_i)$ and
$X_\G(v_k)$.
By Property~\ref{prop:cycle}, there is a cycle, a contradiction. 
\end{proof}

\subsubsection{LSVR implies Condition~\eqref{sw}}

\label{sec:Necessity}
In this section we prove that 
if an LSVR of $\pvph$ exists, then condition~\eqref{sw} holds for at
least one of the four drawings \SW ,\NW, \SE, or \NE. The following
two lemmas allow us to concentrate only on canonical LSVRs.

\begin{restatable}{lemma}{monotone}
\label{lem:monotone}
If $\pvph$ has an LSVR then it has a canonical LSVR. 
\end{restatable}
\begin{proof}
We begin by demonstrating that:  

\begin{claim}
\label{claim:noncollinear}
If $\pvph$ has an LSVR then it has a noncollinear LSVR.
\end{claim}

\begin{proof}
We first show how to transform an LSVR $\G$ of $\pvph$ into an LSVR $\Gp$ such that $\rp(i)\neq\lp(j)$ and $\bp(i)\neq\tp(j)$ for all $i,j\in V$.
Suppose that $\r(i) = \l(j) = x_0$ for some $i,j\in V$.
Let $L=\{k\in V:\r(k)\leq x_0\}$ and $R=\{k\in V:\l(k)\geq x_0\}$ be the collection of shapes to the left and right of $x_0$ respectively. Let $\delta$ be a positive number. Construct $\Gp$ as follows. 

Shift $\G(L)$ to the left by $\delta$, and $\G(R)$ to the right by $\delta$. For all $\ell\notin L\cup R$ (meaning that $\l(\ell)<x_0<\r(\ell)$), stretch $\G(\ell)$ both left and right by $\delta$ (so that $\lp(\ell)=\l(\ell)-\delta$ and $\rp(\ell)=\r(\ell)+\delta$). Note that $\rp(i)<\lp(j)$.

We claim that $\Gp$ is an LSVR of $\pvph$. Since there was no vertical displacement, the horizontal visibilities present in $\G$ are unaffected. Moreover, the structure of the drawing to the right of (and including) $x_0$ in $\G$ was unchanged in $\Gp$. Formally, $\Gp\cap [x_0+\delta,\infty)\times(-\infty,\infty)$ is precisely $\G\cap[x_0,\infty)\times(-\infty,\infty)$ shifted by $\delta$.
Similarly for the structure to the left of $x_0$ in $\G$.
Since $\Ghp(v)$ lies outside the region $M=(x_0-\delta,x_0+\delta)\times (-\infty,\infty)$ for all $v \in V$, we introduce no crossings.

It remains only to show that no unwanted vertical visibilities are introduced in $M$. Notice that any unwanted vertical visibilities in this region must be among shapes not in $\Gp(L \cup R)$ (i.e., those which were stretched), since $\Gp(L\cup R)\cap  M=\emptyset$. If $\Gvp(u)$ and $\Gvp(v)$ share an unwanted vertical visibility in $M$
then, since the visibility was blocked by two horizontal bars, say $\Gv(i')$ and $\Gv(j')$, that shared a collinearity at $x_0$ in $\G$, $\projx(u)$ and $\projx(v)$ both strictly contain $x_0$ and $\Gv(u)$ is above $\Gv(i')$ and $\Gv(j')$ while $\Gv(v)$ is below them. However, this implies a cycle by Property~\ref{prop:cycle}, a contradiction. This demonstrates that $\Gp$ is an LSVR of $\pvph$.

If instead $\bp(i)=\tp(j)$, we perform a similar transformation, but with the geometry rotated by $\pi/2$.
Repeating this process produces
%
an LSVR such that $\r(i)\neq\l(j)$ and $\b(i)\neq\t(j)$ for all $i,j\in V$.

To remove any remaining collinearity,
suppose $\r(i)=\r(j)=x_0$ for some $i,j \in V$ (the other cases $\b(i)=\b(j)$, $\t(i)=\t(j)$, and $\l(i)=\l(j)$ are symmetric) 
and choose $i$ and $j$ such that $\Gv(i)$ is the highest such bar and $\Gv(j)$ is the lowest.
By Property~\ref{prop:cycle}, there cannot be a bar $\Gv(u)$ above $\Gv(i)$ and a bar $\Gv(v)$ below $\Gv(j)$ such that $x_0 \in \projx(u)$ and $x_0 \in \projx(v)$.
If such a $\Gv(u)$ does not exist, decrease $\r(i)$ by a small amount.
This does not introduce any new vertical visibility since $\Gv(i)$ is the highest bar with $x_0 \in \projx(i)$.
Similarly, decrease $\r(j)$ if such a $\Gv(v)$ does not exist.
Repeating this process iteratively produces a noncollinear LSVR of $\pvph$.
\end{proof}

We now proceed to the main proof of Lemma~\ref{lem:monotone}. 
By Claim~\ref{claim:noncollinear}, we may assume that $\G$ is noncollinear.
Suppose without loss of generality that $\Gv(\{\pi_1, \ldots, \pi_{k-1}\})$ is monotone increasing for some $k \geq 3$.
If $\r(\pi_{k}) < \r(\pi_{k-1})$, we will show how to modify $\G$ such that $\Gv(\{\pi_1,\dots,\pi_k\})$ is monotonically increasing or decreasing.
We consider three cases:

Case 1:
$\r(\pi_{k-2}) < \l(\pi_k)$.
There is no vertex $u$ such that $ \r(\pi_k) \in X_\G(u)$ as otherwise, by Property~\ref{prop:ThinOverlap} (with $x=\r(\pi_k)+\epsilon$),
there is a path from $u$ to $\pi_{k-1}$ avoiding $\pi_k$ and $\pi_{k-2}$,
implying that $\pi_{k-1}$ has degree three in $P_\V$; a contradiction. Hence, we may stretch $\Gv(\pi_k)$ to the right such that $\r(\pi_k)>\r(\pi_{k-1})$, making $\Gv(\{\pi_1,\dots,\pi_k\})$ monotonically increasing. 

Case 2:
$\r(\pi_k) < \l(\pi_{k-2})$.
Since $\Gv(\pi_k)$ and $\Gv(\pi_{k-1})$ must share a visibility, we have that $\l(\pi_{k-1})<\l(\pi_{k-2})$, hence $X_\G(\pi_{k-2}) \subset X_\G(\pi_{k-1})$.
By Property~\ref{prop:tower}(i), $X_\G(\pi_1) \subset X_\G(\pi_2) \subset \dots \subset X_\G(\pi_{k-1})$.
Successively, for $j$ from $k-2$ to 1, stretch $\Gv(\pi_j)$ to the right such that $\r(\pi_j)>\r(\pi_{j+1})$. Note that before each stretch, $\r(\pi_{j+1})$ is the rightmost point in $\Gv(\{\pi_1,\dots,\pi_k\})$; otherwise, as in Case 1, $\pi_{j+1}$ has degree three.
Thus, the transformation induces no unwanted vertical visibilities. Furthermore, it's clear that the horizontal visibilities are maintained since there is no movement of any vertical bars. $\Gv(\{\pi_1,\dots,\pi_k\})$ is now monotonically decreasing, which completes the proof of this case. 

This completes the proof if $\Gv(\{\pi_1,\dots,\pi_{k-1}\}$ is monotonically increasing. If $\Gv(\{\pi_1, \dots,\pi_{k-1}\})$ is monotonically decreasing, Property~\ref{prop:cycle}
(for $\pi_k$, $\pi_{k-1}$, and $\pi_{k-2}$)
implies a cycle if $\r(\pi_k) > \r(\pi_{k-1})$. Therefore, $\G_\V(\{\pi_1,\dots,\pi_k\})$ is monotonically decreasing. A similar and symmetric argument may be applied to $\Gh$.
\end{proof}

\comment{
Observe that in the transformations depicted in Fig.~\ref{fig:transformations}, $\Gv(\{i,i+1\})$ and $\Gh(\{i,i+1\})$ are altered from being strictly increasing to simply monotonically increasing. In order to determine when two L-shapes can be ``uncrossed,'' therefore, we first determine in which parts of the drawing $\Gh$ and $\Gv$ are required to be strictly increasing or decreasing. 
}


The preceding lemmas establish that if $\pvph$ has an LSVR then it has a canonical LSVR.  The following lemma connects the existence of a canonical LSVR to the condition that allows the \lsvrpaths Algorithm to stretch its initial drawing to eliminate crossings.
If the algorithm is unable to eliminate these crossings then no canonical LSVR for $\pvph$ exists and hence no LSVR exists.

\begin{restatable}{lemma}{strict}
\label{lem:strict}
Suppose $\G$ is a canonical LSVR of $\pvph$.
If $\Gv$ is monotonically increasing
then $\Gv(\{\pi_c,\pi_{c+1},\dots,\pi_n\})$
is strictly increasing.
Similarly, if $\Gh$ is monotonically increasing 
then $\Gh(\{a,a+1,\dots,n\})$
is strictly increasing,
where $a = |\phprim|$
and $c = |\pvprim|$.
\end{restatable}
\begin{proof}
Let $\G$ be a canonical BVR of a path $P=(v_1,\dots,v_n)$.
We begin by showing: 

\begin{claim}If $\G(\{v_i,v_{i+1}\})$  is strictly increasing (resp., decreasing) then\\ $\G(\{v_i,v_{i+1},\ldots,v_n\})$  is strictly increasing (resp., decreasing).
\end{claim}

\begin{proof}
Suppose that $\G(\{v_i,v_{i+1}\})$ is strictly increasing (the case of strictly decreasing is similar) and let $j> i$ be minimal such that $\G(\{v_j,v_{j+1}\})$ is not.
Since $\G$ is monotone and noncollinear, this implies $\l(v_{j+1}) < \l(v_j)$ and $\projx(v_{j})\subset\projx(v_{j+1})$.
Suppose that $\G(v_{j-1})$ is above $\G(v_{j})$; the other case is argued similarly.  If $\G(v_{j+1})$ is above $\G(v_{j-1})$ or below $\G(v_{j})$, applying Property~\ref{prop:cycle} to $v_{j-1},v_{j}$ and $v_{j+1}$ yields a contradiction. On the other hand, if $\G(v_{j+1})$ is between $\G(v_{j-1})$ and $\G(v_j)$, then because $\projx(v_{j})\subset\projx(v_{j+1})$, $\G(v_{j-1})$ and $\G(v_{j})$ cannot share a visibility. This is a contradiction. 
\end{proof}

We now show that: 

\begin{claim}
If $\G$ is a canonical LSVR of $\pvph$ and $\Gv$ is monotonically increasing
then $\Gv(\{\pi_c,\pi_{c+1}\})$
is strictly increasing.
Similarly, if $\Gh$ is monotonically increasing 
then $\Gh(\{a,a+1\})$
is strictly increasing.
\end{claim}

\begin{proof}
Consider $\Gv$, the proof for $\Gh$ can be obtained by a symmetric argument. We will assume that both $\Gv$ and $\Gh$ are monotonically increasing.
Assume for contradiction that $\Gv(\{\pi_c,\pi_{c+1}\})$ is not strictly increasing, so $\l(\pi_{c+1})<\l(\pi_{c})$. Notice that, since $\Gv$ is monotonically increasing, $ X_\G(\pi_{c})\subset X_\G(\pi_{c+1})$. Applying Property~\ref{prop:tower}(i) we see that $X_\G(\pi_i)\subset X_\G(\pi_{i+1})$ for $i \in [c]$. 
We consider two cases based on the ordering of $(\pi_1,\dots,\pi_c)$ in $\ph$ (see Figure~\ref{fig:strict}). 

\begin{figure}
\centering
\begin{subfigure}[t]{0.2\textwidth}
\includegraphics[width=\textwidth]{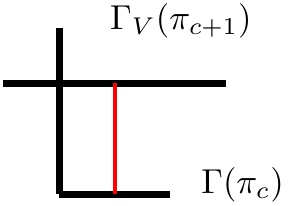}
\caption{}
\end{subfigure}
\hspace{0.2cm}
\begin{subfigure}[t]{0.2\textwidth}
\includegraphics[width=\textwidth]{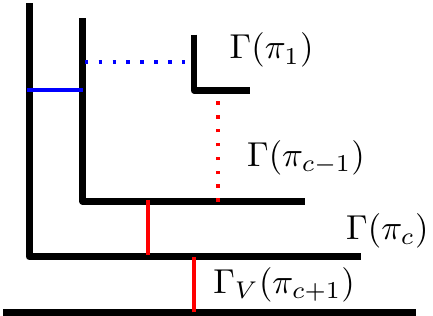}
\caption{}
\end{subfigure}
\hspace{0.2cm}
\begin{subfigure}[t]{0.2\textwidth}
\includegraphics[width=\textwidth]{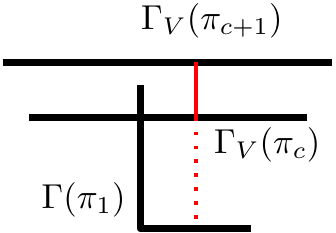}
\caption{}
\end{subfigure}
\hspace{0.2cm}
\begin{subfigure}[t]{0.2\textwidth}
\includegraphics[width=\textwidth]{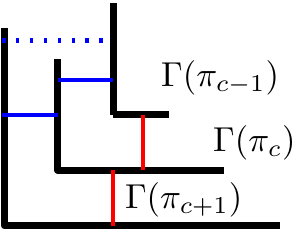}
\caption{}
\end{subfigure}
\caption{The cases in the proof of Lemma \ref{lem:strict}. (a) and (b) correspond to the two subcases of Case 1; (b) and (d) to the two subcases of Case 2. }
\label{fig:strict}
\end{figure}

Case 1: 
Suppose $\pi_1<\pi_2<\dots<\pi_c$. Then $\t(\pi_1)<\t(\pi_2)<\dots<\t(\pi_c)$ since $(\pi_1, \pi_2, \dots, \pi_c)$ is increasing in $P_\H$ and $\Gh$ is monotonically increasing.
By definition of $\pvprim$, $\pi_{c+1}<\pi_{c}$, hence $\t(\pi_{c+1})<\t(\pi_{c})$.
Therefore, $\Gv(\pi_c)$ is above $\Gv(\pi_{c+1})$ (i.e., $\b(\pi_c) > \b(\pi_{c+1})$);
otherwise $\Gh(\pi_c)$ would intersect $\Gv(\pi_{c+1})$ (by the assumption that $\l(\pi_{c+1})<\l(\pi_{c})$).
Hence, by Property~\ref{prop:tower}(ii), $\b(\pi_1)>\b(\pi_2)>\dots>\b(\pi_{c+1})$ thus demonstrating that $\G(\pi_i)$ is nested in $\G(\pi_{i+1})$ for all $i\in[c-1]$. 

First, we claim that for all $i\in\{\pi_1, \pi_1 + 1, \dots, \pi_c \}$, $\l(i)\in[\l(\pi_c),\l(\pi_1)]$.
Suppose the claim is false and choose the smallest $i\in\{\pi_1, \pi_1 + 1, \dots,\pi_c\}$ such that
either $\l(i)<\l(\pi_c)$ or $\l(i)>\l(\pi_1)$. In the former case, $\Gh(i)$ is blocked from sharing a visibility with $\Gh(i-1)$ by $\Gh(\pi_c)$
(where we're using that $\t(i-1),\t(i)<\t(\pi_c)$, and $\b(i-1) > \b(\pi_c)$ to avoid $\Gh(i-1)$ and $\Gv(\pi_c)$ intersecting).
In the latter case apply Property~\ref{prop:cycle} to $i,\pi_1$ and $\pi_c$ (using that $\t(\pi_c)>\t(i)>\t(\pi_1)$) to obtain a contradiction.
This proves the claim, which implies that $\pi_{c+1}<\pi_1$, since $\pi_{c+1} < \pi_c$ and $\l(\pi_{c+1})\notin[\l(\pi_c), \l(\pi_1)]$.
This, however, also yields a contradiction: Since $\projy(\pi_1)\subset\projy(\pi_{c})$ and $\l(\pi_{c})<\l(\pi_1)$, Property~\ref{prop:tower} implies that $\l(i)>\l(\pi_c)$ for all $i<\pi_1$. By assumption however, $\b(\pi_{c+1})<\b(\pi_c)$. 

Case 2: 
Otherwise, $\pi_1>\pi_2>\dots>\pi_c$ so $\t(\pi_1)>\t(\pi_2)>\dots>\t(\pi_c)$. If $\b(\pi_{c+1})<\b(\pi_c)$, then we must have $c=1$;
otherwise apply Property~\ref{prop:cycle} to $\pi_{c+1},\pi_{c-1}$ and $\pi_c$. 
This, however, contradicts the definition of $\pvprim$ as it will always have length at least two.
Thus $\b(\pi_c)<\b(\pi_{c+1})$ and again by Property~\ref{prop:tower}(ii) we have $\b(\pi_1)<\dots<\b(\pi_{c+1})$.
Again, however, this forces $c=1$. Otherwise, because $\l(\pi_c)<\l(\pi_1)$ and $\t(\pi_1)>\t(\pi_c)$, $\Gh(\pi_1)$ would intersect $\Gv(\pi_c)$. As above, $c=1$ is impossible. 

This completes the proof if $\Gv$ is increasing. See Fig.~\ref{fig:strict} for an illustration of the different cases. If $\Gv$ is decreasing, the argument is similar but considers instead the vertices $\pi_{d-1},\pi_d,\dots,\pi_n$. The same geometric arguments apply. 
\end{proof}

Combining these two claims provide the proof of Lemma~\ref{lem:strict}. 
\end{proof}

We can now prove that if an LSVR of $\pvph$ exists then
condition~\eqref{sw} holds for at least one drawing \SW{},\NW{},\SE{},
or \NE{}, which completes the proof of necessity. 


By Lemma~\ref{lem:monotone}, we may assume that if an LSVR of $\pvph$
exists then there is an LSVR $\G$ that is monotone and noncollinear.
We claim that if condition~\eqref{sw} doesn't hold for
the \SW{} drawing then no LSVR $\G$ exists in which $\Gv$ and $\Gh$
are both monotonically increasing.
In a similar manner, if condition~\eqref{sw} doesn't hold for
the \NW{}, \SE{}, or \NE{} drawing then
no LSVR $\G$ exists in which 
$\Gv$ is increasing and $\Gh$ is decreasing,
$\Gv$ is decreasing and $\Gh$ increasing, or
$\Gv$ and $\Gh$ are both decreasing.
Since these are the only four possibilities for a monotone LSVR, by Lemma \ref{lem:monotone} if none of these four monotonic LSVRs exist, then no LSVR of $\pvph$ exists. 
Suppose $\G$ is an LSVR of $\pvph$ in which $\Gv$ and $\Gh$ are
monotonically increasing and let $i$ be an index where
condition~\eqref{sw} fails: $(i+1,i)\in P_\V$, $i\notin \pvprim$, and
$i+1\notin \phprim$. Let $a=|\phprim|$.
By Lemma \ref{lem:strict}, $\Gh(\{a,a+1,\dots,n\})$ is strictly
increasing. Thus, since $i\geq a$ (because $i+1\notin\phprim$),
$\b(i)<\b(i+1)<\t(i)$. By similar reasoning, we obtain that
$\Gv(\{i+1,i\})$ is strictly increasing, so
$\l(i+1)<\l(i)<\r(i+1)$. However, this is an impossible configuration
to realize without an intersection between $\Gh(i)$ and
$\Gv(i+1)$. Therefore no such LSVR exists.

\section{Conclusion}
We show that deciding if a pair of undirected graphs $\la G_\V, G_\H \ra$,
which share the same vertex set, has a simultaneous visibility representation using unit squares or rectangles is \NP-complete.
The graphs used in the hardness reduction are only slightly more complex than paths.
We also show an efficient algorithm to decide if pairs of paths have simultaneous visibility representations using fixed orientation L-shapes.
This narrows the range of possible graph families for which determining simultaneous visibility representation is non-trivial yet not \NP-hard.
Perhaps the most interesting open questions for this representation involve identifying families of pairs of graphs that lie within this range.
Pairs of graphs with more structure than simple paths are certainly realizable.
Characterizing such pairs is also an interesting open problem.


\small
\bibliographystyle{abbrv}
\bibliography{ref.bib}

\end{document}